\documentclass[lettersize,journal]{IEEEtran}
\usepackage{amsmath,amsfonts}
\usepackage[ruled,lined,linesnumbered,noend]{algorithm2e}
\usepackage{array}
\usepackage{textcomp}
\usepackage{stfloats}
\usepackage{url}
\usepackage{verbatim}
\usepackage{graphicx}
\usepackage{cite}
\usepackage{amsthm}
\usepackage{booktabs,tabularx}
\usepackage[normalem]{ulem}
\usepackage[table]{xcolor}
\usepackage{multirow}
\usepackage{subfigure}
\usepackage{tikz}
\usetikzlibrary{positioning, shapes, calc, arrows.meta, decorations.pathreplacing}
\usepackage{pgfplots}

\newtheorem{theorem}{Theorem}[section]

\newtheorem{definition}{Definition}

\begin{document}

\title{Top-k Approximate Functional Dependency Discovery}

\author{Xiaolong~Wan, Xixian~Han
\thanks{The authors are with School of Computer Science and Technology, Harbin Institute of Technology, China. (e-mail: wxl@hit.edu.cn, hanxx@hit.edu.cn)}
\thanks{}}

\markboth{}%
{Shell \MakeLowercase{\textit{et al.}}: A Sample Article Using IEEEtran.cls for IEEE Journals}


\maketitle

\begin{abstract}
Approximate functional dependencies (AFDs) relax exact functional dependencies by tolerating a bounded degree of violation, making them suited for data quality auditing. Threshold-based discovery returns all dependencies above a user-specified cutoff, but output size is uncontrollable, the right threshold varies across datasets, and widely used measures are sensitive to LHS dimensionality. We study global top-$k$ AFD discovery, where neither the LHS nor the RHS is fixed and the $k$ strongest dependencies under $\mu^+$ are returned directly. The cross-attribute comparability of $\mu^+$ makes such a global ranking well-defined. We prove a Triangle Incompatibility Theorem showing that minimality, global top-$k$ ranking, and exact-$k$ output cannot simultaneously hold under any non-monotonic scoring function, justifying the removal of the minimality requirement. We present two algorithms: TALE-Base, which returns the exact global top-$k$ result by exhaustive level-wise evaluation, and TALE-Opt, which reduces computation through Apriori-style candidate generation, LHS computation reuse, and two complementary pruning rules exploiting exact FD monotonicity and an optimistic upper bound on $\mu^+$. Experiments on 41 real-world datasets show that TALE-Opt achieves pruning ratios up to 99.81\% and speedups over TALE-Base up to 78.81$\times$.
\end{abstract}


\section{Introduction} \label{sec:intro}

\IEEEPARstart{F}{unctional} dependencies (FDs) are among the most studied integrity constraints in relational databases~\cite{DBLP:books/mg/SKS20}. An FD $X \to A$ states that the values of $X$ uniquely determine $A$, and discovering FDs from data has long served as a basis for schema design~\cite{DBLP:journals/is/KohlerL18}, data cleaning~\cite{DBLP:journals/jdiq/BoecklingB25}, and consistency checking~\cite{DBLP:conf/sigmod/ZhangGR20}. In practice, strict FDs seldom hold over real-world datasets: dirty values, recording errors, and legitimate exceptions all introduce violations. Approximate functional dependencies (AFDs)~\cite{DBLP:journals/pvldb/0001N18} address this by tolerating a bounded degree of violation. For data quality auditing, an AFD that is \emph{almost} satisfied is often more useful than an exact FD, since it points directly to the tuples responsible for the violation and supports targeted error detection and repair~\cite{DBLP:books/acm/IlyasC19}.

\begin{figure}[t]
	\centering
	\begin{tikzpicture}[
		x=1cm, y=1cm,
		dot/.style={circle, inner sep=1.5pt, fill},
		dotgray/.style={circle, inner sep=1.5pt, fill=black!20},
		highlight/.style={circle, inner sep=1.5pt, fill, draw=red, line width=0.8pt},
		brace/.style={decorate, decoration={brace, amplitude=3pt}},
		]
		
		\def\axislen{5.5}
		\def\rowsep{2.2}
		
		\begin{scope}[shift={(0, 2*\rowsep)}]
			\node[font=\small\bfseries, anchor=west] at (0, 0.8) {(a) Threshold-based};
			
			\draw[->, thick] (0, 0) -- (\axislen+0.3, 0);
			\node[font=\scriptsize, anchor=north] at (0, -0.05) {0.5};
			\node[font=\scriptsize, anchor=north] at (\axislen, -0.05) {1};
			\node[font=\scriptsize, anchor=north west] at (\axislen+0.15, -0.05) {$\mathit{score}$};
			
			\node[dot, fill=blue!70] at (0.94*\axislen, 0) {};        
			\node[dot, fill=orange!80] at (0.88*\axislen, 0) {};       
			\node[dot, fill=blue!70] at (0.82*\axislen, 0) {};         
			\node[dot, fill=orange!80] at (0.76*\axislen, 0) {};       
			\node[dot, fill=blue!70] at (0.70*\axislen, 0) {};         
			\node[dot, fill=green!60!black] at (0.64*\axislen, 0) {};  
			\node[dot, fill=red!60] at (0.58*\axislen, 0) {};          
			\node[dot, fill=orange!80] at (0.52*\axislen, 0) {};       
			\node[dot, fill=green!60!black] at (0.44*\axislen, 0) {};  
			\node[dot, fill=purple!60] at (0.36*\axislen, 0) {};       
			
			\draw[dashed, thick, red!70!black] (0.80*\axislen, -0.35) -- (0.80*\axislen, 0.5);
			\node[font=\scriptsize, red!70!black, anchor=south west] at (0.80*\axislen+0.05, 0.4) {$\theta{=}0.90$};
			
			\node[font=\scriptsize, text width=5.5cm, anchor=north west] at (0, -0.45) {%
				Output size depends on $\theta$: $\theta{=}0.80 \Rightarrow 7$;\; $\theta{=}0.95 \Rightarrow 1$;\; $\theta{=}0.90 \Rightarrow 3$.};
		\end{scope}
		
		\begin{scope}[shift={(0, \rowsep)}]
			\node[font=\small\bfseries, anchor=west] at (0, 0.8) {(b) Fixed-RHS top-$k$ \normalfont\scriptsize (RHS $= D$)};
			
			\draw[->, thick] (0, 0) -- (\axislen+0.3, 0);
			\node[font=\scriptsize, anchor=north] at (0, -0.05) {0.5};
			\node[font=\scriptsize, anchor=north] at (\axislen, -0.05) {1};
			\node[font=\scriptsize, anchor=north west] at (\axislen+0.15, -0.05) {$\mathit{score}$};
			
			\node[dotgray] at (0.94*\axislen, 0) {};       
			\node[dotgray] at (0.82*\axislen, 0) {};        
			\node[dotgray] at (0.70*\axislen, 0) {};        
			\node[dotgray] at (0.64*\axislen, 0) {};        
			\node[dotgray] at (0.58*\axislen, 0) {};        
			\node[dotgray] at (0.44*\axislen, 0) {};        
			\node[dotgray] at (0.36*\axislen, 0) {};        
			
			\node[highlight, fill=orange!80] at (0.88*\axislen, 0) {};  
			\node[highlight, fill=orange!80] at (0.76*\axislen, 0) {};  
			\node[highlight, fill=orange!80] at (0.52*\axislen, 0) {};  
			
			\node[font=\scriptsize, text width=5.5cm, anchor=north west] at (0, -0.45) {%
				Only searches $X {\to} D$. Misses $\{B,C\}{\to}E$ ($score{=}0.97$).};
		\end{scope}
		
		\begin{scope}[shift={(0, 0)}]
			\node[font=\small\bfseries, anchor=west] at (0, 0.8) {(c) Global top-$k$ \normalfont\scriptsize (this paper)};
			
			\draw[->, thick] (0, 0) -- (\axislen+0.3, 0);
			\node[font=\scriptsize, anchor=north] at (0, -0.05) {0.5};
			\node[font=\scriptsize, anchor=north] at (\axislen, -0.05) {1};
			\node[font=\scriptsize, anchor=north west] at (\axislen+0.15, -0.05) {$\mathit{score}$};
			
			\node[dot, fill=orange!80] at (0.76*\axislen, 0) {};       
			\node[dot, fill=blue!70] at (0.70*\axislen, 0) {};         
			\node[dot, fill=green!60!black] at (0.64*\axislen, 0) {};  
			\node[dot, fill=red!60] at (0.58*\axislen, 0) {};          
			\node[dot, fill=orange!80] at (0.52*\axislen, 0) {};       
			\node[dot, fill=green!60!black] at (0.44*\axislen, 0) {};  
			\node[dot, fill=purple!60] at (0.36*\axislen, 0) {};       
			
			\node[highlight, fill=blue!70] at (0.94*\axislen, 0) {};    
			\node[highlight, fill=orange!80] at (0.88*\axislen, 0) {};  
			\node[highlight, fill=blue!70] at (0.82*\axislen, 0) {};    
			
			\draw[brace, thick] (0.95*\axislen, 0.2) -- (0.81*\axislen, 0.2);
			\node[font=\scriptsize, anchor=south] at (0.88*\axislen, 0.35) {top-3};
			
			\node[font=\scriptsize, text width=5.5cm, anchor=north west] at (0, -0.45) {%
				Searches all $X {\to} A$. No $\theta$, no fixed RHS. Returns the 3 globally strongest AFDs.};
		\end{scope}
	\end{tikzpicture}
	
	\vspace{2pt}
	\par\noindent
	\centering
	{\scriptsize \textit{RHS}:\;\;
		\tikz[baseline=-0.5ex]\draw[fill=blue!70] (0,0) circle (1.5pt);\ $E$\quad
		\tikz[baseline=-0.5ex]\draw[fill=orange!80] (0,0) circle (1.5pt);\ $D$\quad
		\tikz[baseline=-0.5ex]\draw[fill=green!60!black] (0,0) circle (1.5pt);\ $C$\quad
		\tikz[baseline=-0.5ex]\draw[fill=red!60] (0,0) circle (1.5pt);\ $A$\quad
		\tikz[baseline=-0.5ex]\draw[fill=purple!60] (0,0) circle (1.5pt);\ $B$\quad
		\tikz[baseline=-0.5ex]\draw[fill=black!20] (0,0) circle (1.5pt);\ ignored}

	\caption{Comparison of three AFD discovery paradigms on a relation with attributes $\{A, B, C, D, E\}$. Each dot represents a candidate AFD; colors indicate different RHS attributes (see legend). (a)~The threshold-based approach returns all AFDs above $\theta$, but output size is sensitive to $\theta$. (b)~The fixed-RHS approach restricts the search to a single target, missing high-scoring AFDs with other RHS. (c)~The global top-$k$ approach returns the $k$ strongest AFDs across all attributes without any pre-specification.}
	\label{fig:motivation}
\end{figure}
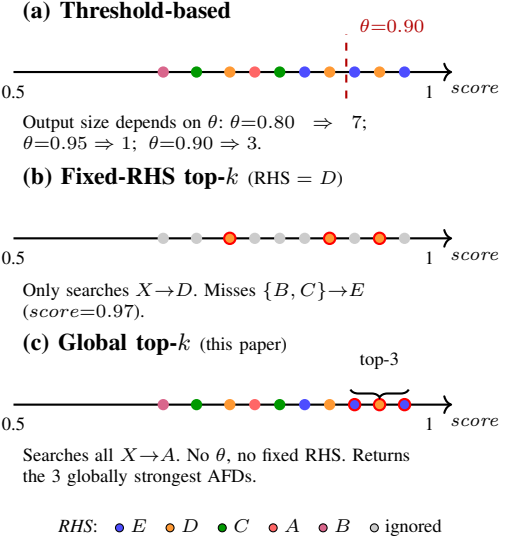

Existing approaches to AFD discovery predominantly follow a threshold-based paradigm: a user specifies a cutoff $\theta$ and all dependencies whose measure exceeds $\theta$ are returned~\cite{DBLP:journals/tkde/CaruccioDP16}. Three difficulties undermine this paradigm in practice.
\begin{itemize}
	\item \textit{Uncontrollable output size.} The number of returned dependencies depends entirely on $\theta$. A high cutoff risks missing genuine but imperfect dependencies, while a low one floods the output with spurious results~\cite{DBLP:conf/kdd/MandrosBV17}. There is no principled way to choose between these extremes.
	
	\item \textit{Data-dependent threshold selection.} No single $\theta$ works universally. The appropriate value depends on the data distribution, attribute cardinalities, and noise level, and it varies not only across datasets but across different attribute combinations within the same dataset.
	
	\item \textit{Measure sensitivity to LHS dimensionality.} Widely used measures such as $g_3$ are sensitive to the number of distinct values on the left-hand side (LHS-uniqueness)~\cite{DBLP:journals/vldb/ParciakWHNPV25}. As the LHS grows, even random attribute combinations tend to score high simply because fine-grained partitioning leaves few tuples per group. The same numerical threshold therefore carries different meanings at different LHS sizes, and this dimensionality bias compounds the first two difficulties.
\end{itemize}
Figure~\ref{fig:motivation} illustrates these problems and previews how the global top-$k$ approach proposed in this paper addresses them.

The top-$k$ approach avoids threshold selection by returning the $k$ highest-scoring dependencies directly. We extend this idea to a global setting where neither the LHS nor the RHS is fixed. Mandros et al.~\cite{DBLP:conf/kdd/MandrosBV17} study the fixed-RHS case, searching over subsets of the remaining attributes to find those that best determine a given target $A$. Their scoring function is a corrected fraction of information based on Shannon entropy, normalized by $H(A)$. Since the normalization by $H(A)$ depends on the marginal entropy of each RHS attribute, scores for dependencies with different right-hand sides reflect different reference points and are not directly comparable, so running their method separately for each possible target and merging the results does not yield a meaningful global ranking. A cross-RHS comparable scoring function is therefore a prerequisite for global top-$k$ discovery, a requirement that Shannon entropy-based measures cannot satisfy.

We address the problem of \emph{global} top-$k$ AFD discovery and adopt the $\mu^+$ measure as the scoring function. $\mu^+$ is grounded in logical entropy and quantifies the probability that two randomly chosen tuples agreeing on $X$ also agree on $A$; its correction factor has a closed-form expression, keeping computation practical for the larger candidate space of global discovery. A recent comparative study~\cite{DBLP:journals/vldb/ParciakWHNPV25} shows that $\mu^+$ is robust to both LHS-uniqueness and RHS-skew. Robustness to RHS-skew is especially important here, since a measure sensitive to it would systematically favor skewed attributes and flood the top-$k$ list with artifacts rather than genuine dependencies. 

Our formulation imposes no minimality constraint on the LHS, not as a simplification, but because minimality is theoretically incompatible with global top-$k$ ranking and exact-$k$ output under any non-monotonic scoring function, as we prove in Theorem~\ref{thm:triangle}.

The global top-$k$ AFD problem requires searching over all $X \to A$ combinations up to a bounded LHS size, a space that grows rapidly with the number of attributes. In the fixed-RHS setting of Mandros et al.~\cite{DBLP:conf/kdd/MandrosBV17}, branch-and-bound exploits a single optimistic estimator over all LHS subsets. The global setting has no such structure, since no single bounding function can guide pruning across different RHS attributes. We propose two algorithms under the name TALE (\textit{T}op-k \textit{A}FD discovery via \textit{L}ogical \textit{E}ntropy). TALE-Base adopts a level-wise enumeration organized by LHS size, maintaining a min-heap of the current top-$k$ results, and guarantees the exact global top-$k$ result by exhaustive evaluation. TALE-Opt reduces computation through techniques at three levels. At the search level, an Apriori-style candidate generation scheme propagates pruning decisions to descendant LHS candidates. At the computation level, LHS hash values and NULL flags are cached and reused across all RHS evaluations sharing the same LHS. At the pruning level, exact FD pruning eliminates attributes whose supersets are guaranteed exact. The central challenge is that $\mu^+$ is not monotone under LHS enlargement, which rules out classical anti-monotonic pruning for the remaining candidates. We address this by deriving an optimistic upper bound on $\mu^+$ that decreases monotonically with LHS size, recovering safe pruning despite the non-monotonicity of the measure.

This paper makes the following contributions.
\begin{itemize}
	\item The problem of \emph{global top-$k$ AFD discovery} is formalized, where neither the LHS nor the RHS is fixed. Adopting the $\mu^+$ measure makes dependencies with different RHS attributes directly comparable, giving a global ranking a well-defined meaning. To our knowledge, this problem has not been previously studied.
	
	\item A \emph{Triangle Incompatibility Theorem} shows that minimality, global top-$k$ ranking, and exact-$k$ output cannot simultaneously hold when the scoring function is not monotone under LHS enlargement, providing theoretical justification for dropping the minimality requirement.
	
	\item A baseline algorithm TALE-Base and an optimized algorithm TALE-Opt are presented. TALE-Opt reduces computation through Apriori-style candidate generation, LHS computation reuse, exact FD pruning, and an optimistic upper bound on $\mu^+$ that recovers safe pruning under non-monotonic scoring.
	
	\item Experiments on 41 real-world datasets show that TALE-Opt achieves pruning ratios up to 99.81\% and speedups over TALE-Base up to 78.81$\times$.
\end{itemize}

The rest of this paper is organized as follows. Section~\ref{sec:preliminary} introduces preliminaries. Section~\ref{sec:incompatibility} proves the Triangle Incompatibility Theorem. Section~\ref{sec:baseline} presents the baseline algorithm, which is optimized in Section~\ref{sec:optimized}. Section~\ref{sec:experiments} reports experimental results. Section~\ref{sec:related} discusses related work, and Section~\ref{sec:conclusion} concludes the paper.

\section{Preliminaries} \label{sec:preliminary}

\subsection{Basic Notions}

Let $R = \{A_1, A_2, \ldots, A_m\}$ be a relation schema of $m$ attributes, and let $r$ be a relation instance over $R$ containing $n$ tuples. For a tuple $t \in r$ and an attribute $A \in R$, $t[A]$ denotes the value of $A$ in $t$. For a subset $X \subseteq R$, $t[X]$ denotes the projection of $t$ onto $X$. We denote by $\pi_X$ the projection of $r$ onto $X$, and by $d_X = |\pi_X|$ the number of distinct $X$-values in $r$.

\begin{definition}[Functional Dependency]
	Given $X \subseteq R$ and $A \in R \setminus X$, a functional dependency (FD) $X \to A$ holds in $r$ if and only if for all $t_1, t_2 \in r$, $t_1[X] = t_2[X]$ implies $t_1[A] = t_2[A]$.
\end{definition}

An FD $X \to A$ is \emph{minimal} if no proper subset $X' \subset X$ also satisfies $X' \to A$ in $r$. When an FD does not hold exactly, a scoring function quantifies how close it is to holding. A dependency scored in this way is called an \emph{approximate functional dependency} (AFD).

\begin{definition}[Approximate Functional Dependency] \label{def:afd}
	An approximate functional dependency (AFD) over a schema $R$ is a pair $(X \to A, \sigma)$, where $X \subseteq R$, $A \in R \setminus X$, and $\sigma$ is a scoring function that maps each dependency and relation instance to a value in $[0, 1]$. For a given instance $r$, the value $\sigma(X \to A, r)$ quantifies the strength of the dependency, with $\sigma = 1$ if and only if $X \to A$ is an exact FD in $r$.
\end{definition}

Many scoring functions have been proposed for AFDs, each with different trade-offs in reliability and sensitivity~\cite{DBLP:journals/vldb/ParciakWHNPV25}. The measure adopted in this paper, $\mu^+$, is defined in Section~\ref{sec:muplus}.

\subsection{The $\mu^+$ Measure}
\label{sec:muplus}

We adopt the $\mu^+$ measure~\cite{DBLP:journals/vldb/ParciakWHNPV25} as the scoring function $\sigma$ in Definition~\ref{def:afd}. It is based on logical entropy and built on \emph{probabilistic dependence} (pdep), which measures the probability that two tuples drawn uniformly at random with the same $X$-value also agree on $A$.

\begin{definition}[Probabilistic Dependence] \label{def:pdep}
	Given $X \subseteq R$, $A \in R \setminus X$, and a relation instance $r$ of size $n$, group the tuples by their $X$-values. For each group $g$, let $|g|$ be its size and let $f_{g,a}$ be the number of tuples in $g$ with value $a$ on attribute $A$. The probabilistic dependence of $A$ on $X$ is
	\begin{equation*}
		\mathrm{pdep}(X \to A, r) = \sum_{g} \frac{|g|}{n} \sum_{a} \left( \frac{f_{g,a}}{|g|} \right)^2.
	\end{equation*}
\end{definition}

The inner sum $\sum_a (f_{g,a} / |g|)^2$ is the probability that two tuples chosen uniformly at random from group $g$ agree on $A$. The outer sum averages this over all groups, weighted by relative size. When $X \to A$ is an exact FD, each group is homogeneous in $A$ and $\mathrm{pdep} = 1$. When $X$ provides no information about $A$, pdep reduces to $\sum_a p_a^2$, where $p_a$ is the marginal frequency of value $a$.

The unconditional variant removes the conditioning on $X$:
\begin{equation*}
	\mathrm{pdep}(A, r) = \sum_{a} \left( \frac{|\{t \in r : t[A] = a\}|}{n} \right)^2
\end{equation*}
which is simply the probability that two tuples chosen at random from $r$ agree on $A$. The ratio $\mathrm{pdep}(X \to A) / \mathrm{pdep}(A)$ then measures how much conditioning on $X$ increases agreement on $A$ beyond what the marginal distribution already provides.

\begin{definition}[The $\mu^+$ Measure] \label{def:muplus}
	Given $X \subseteq R$, $A \in R \setminus X$, and a relation instance $r$ of size $n$ with $d_X < n$ and $\mathrm{pdep}(A, r) < 1$, the $\mu^+$ measure is defined as
	\begin{equation*}
		\mu^+(X \to A, r) = \max\!\Big(0,\; 1 - \rho(X \to A, r)\Big)
	\end{equation*}
	where
	\begin{equation*}
		\rho(X \to A, r) = \frac{1 - \mathrm{pdep}(X \to A, r)}{1 - \mathrm{pdep}(A, r)} \cdot \frac{n - 1}{n - d_X}.
	\end{equation*}
	When $d_X = n$ or $\mathrm{pdep}(A, r) = 1$, $X \to A$ is an exact FD and $\mu^+(X \to A, r) := 1$.
\end{definition}

The ratio $(1 - \mathrm{pdep}(X \to A)) / (1 - \mathrm{pdep}(A))$ is the fraction of disagreement on $A$ that remains after conditioning on $X$. A value of zero means $X$ fully determines $A$, while a value of one means $X$ provides no benefit over the marginal. The factor $(n-1)/(n-d_X)$ corrects for a bias that arises when $d_X$ is large relative to $n$: with many small groups, even independent attributes show high pdep, and this factor penalizes that effect. Together, $\rho$ close to zero yields $\mu^+$ close to one, while $\rho \geq 1$ is truncated to $\mu^+ = 0$.

Two properties of $\mu^+$ are relevant here. First, $\mathrm{pdep}(A, r)$ normalizes away the marginal distribution of $A$, so $\mu^+$ scores are directly comparable across dependencies with different right-hand sides, a prerequisite for global top-$k$ discovery. Second, $\mu^+$ ranks among the top two performers across 14 measures in the comparative study of Parciak et al.~\cite{DBLP:journals/vldb/ParciakWHNPV25}, with strong robustness to both LHS-uniqueness and RHS-skew. The other top performer, $\mathrm{RFI}'^+$, achieves comparable ranking quality but requires permutation-based estimation and is orders of magnitude slower. Therefore, Parciak et al.\ recommend $\mu^+$ for practical use.

\subsection{Problem Definition} \label{sec:problem}

We now define the problem studied in this paper, in which neither side of the dependency is fixed.

\begin{definition}[Global Top-$k$ AFD Discovery]
	\label{def:topk}
	Given a relation instance $r$ over schema $R$ and a positive integer $k$, let $\mathcal{F}$ denote the set of all AFDs over $r$ that are not exact FDs:
	\begin{equation*}
		\mathcal{F} = \{ X \to A \mid X \subseteq R,\; A \in R \setminus X,\; \mu^+(X \to A, r) < 1 \}.
	\end{equation*}
	The global top-$k$ AFD discovery problem is to find a subset $\mathcal{F}_k \subseteq \mathcal{F}$ with $|\mathcal{F}_k| = \min(k, |\mathcal{F}|)$ such that for all $f \in \mathcal{F}_k$ and $f' \in \mathcal{F} \setminus \mathcal{F}_k$,
	\begin{equation*}
		\mu^+(f, r) \geq \mu^+(f', r).
	\end{equation*}
	When multiple AFDs share the same $\mu^+$ score at the boundary, any selection among them is acceptable.
\end{definition}

In principle, $\mathcal{F}$ includes AFDs with LHS of any size up to $m - 1$. In practice, we restrict the search to LHS of size at most $L$, a user-specified parameter. While some large-LHS dependencies may still be meaningful, they are typically harder to inspect and validate. Therefore, $L$ is used as a practical upper bound on the LHS size to control the search space and reduce the burden of human validation. Mandros et al.~\cite{DBLP:conf/kdd/MandrosBV17} observed that the average optimal LHS size across 42 datasets is 4.0. As $|X|$ grows, $d_X$ approaches $n$, and the correction factor $(n-1)/(n-d_X)$ in $\mu^+$ heavily penalizes such candidates, making them unlikely to enter the top-$k$. For any choice of $L$, the algorithm returns the exact top-$k$ among all AFDs with $|X| \leq L$. Increasing $L$ enlarges the search space and may change the result, but the algorithm remains exact within the chosen search depth.

The candidate set $\mathcal{F}$ excludes exact FDs ($\mu^+ = 1$), since including them would fill the top-$k$ list with dependencies that carry no information about data quality issues, leaving genuinely approximate dependencies with no chance to appear in the result. Exact FDs are in any case efficiently handled by existing algorithms~\cite{DBLP:journals/tkde/WanHWL24,DBLP:journals/pacmmod/BleifussPBSN24}.

Our problem definition imposes no minimality constraint on the LHS. In classical FD discovery, minimality is natural because an exact FD $X \to A$ logically implies $X' \to A$ for any $X' \supset X$. For AFDs this implication does not hold: enlarging the LHS may increase or decrease $\mu^+$, so a non-minimal dependency can score genuinely differently from any of its subsets. Section~\ref{sec:incompatibility} shows that dropping minimality is not merely convenient but theoretically necessary.

\section{Incompatibility of Top-$k$ and Minimality}
\label{sec:incompatibility}

We prove that minimality, global top-$k$ ranking, and exact-$k$ output are mutually incompatible under any non-monotonic scoring function.

\subsection{Three Properties}

Consider a relation instance $r$ over schema $R$, the scoring function $\mu^+$, and a positive integer $k$, with $\mathcal{F}$ as in Definition~\ref{def:topk}. A result set $\mathcal{S} \subseteq \mathcal{F}$ should ideally satisfy three properties:
\begin{enumerate}
	\item[\textbf{(P1)}] \textit{Minimality.} Every dependency in $\mathcal{S}$ has a minimal LHS: for all $X \to A \in \mathcal{S}$, there is no $X' \subset X$ such that $X' \to A \in \mathcal{F}$.
	
	\item[\textbf{(P2)}] \textit{Global top-$k$.} The dependencies in $\mathcal{S}$ have the $k$ highest $\mu^+$ scores among all candidates: for all $f \in \mathcal{S}$ and $f' \in \mathcal{F} \setminus \mathcal{S}$, $\mu^+(f, r) \geq \mu^+(f', r)$.
	
	\item[\textbf{(P3)}] \textit{Exact-$k$ output.} The result contains exactly $k$ dependencies: $|\mathcal{S}| = k$.
\end{enumerate}
Each property is individually reasonable, but the following theorem shows they cannot hold simultaneously.

\subsection{The Incompatibility Theorem}

\begin{theorem}[Triangle Incompatibility]
	\label{thm:triangle}
	Let $\sigma$ be a scoring function such that there exist a relation instance $r$, a minimal AFD $X \to A$, and a strict superset $X' \supset X$ with $\sigma(X \to A, r) < \sigma(X' \to A, r)$. Then there exist a relation instance and a value of $k$ such that no subset $\mathcal{S} \subseteq \mathcal{F}$ simultaneously satisfies \emph{(P1)}, \emph{(P2)}, and \emph{(P3)}.
\end{theorem}

\begin{proof}
	By assumption, there exist a relation instance $r$, a minimal AFD $f_1 = X \to A$ with score $s_1 = \sigma(f_1, r)$, and a strict superset $X' \supset X$ such that $f_2 = X' \to A$ has score $s_2 = \sigma(f_2, r) > s_1$. Since $X \subset X'$ and $f_1$ is minimal, $f_2$ is non-minimal. We construct a counterexample with $k = 2$.
	
	Let $r$ be extended with attributes $Y$ and $B$ distinct from those in $X' \cup \{A\}$, and let
	\begin{equation*}
		f_3 = Y \to B, \quad \sigma(f_3, r) = s_3
	\end{equation*}
	be a minimal AFD independent of $f_1$ and $f_2$, with $0 < s_3 < s_1$. This is realizable because $Y$ and $B$ are disjoint from $X' \cup \{A\}$, so their distributions can be chosen independently without affecting $s_1$ or $s_2$.
	
	We show that no subset $\mathcal{S} \subseteq \{f_1, f_2, f_3\}$ of size $2$ satisfies all three properties.
	
	\smallskip
	\noindent\textbf{(P2) + (P3) $\Rightarrow$ $\neg$(P1).}
	The two highest-scoring AFDs are $f_2$ ($\sigma = s_2$) and $f_1$ ($\sigma = s_1$), so the global top-2 is $\{f_2, f_1\}$. But $f_2$ is non-minimal since $X \subset X'$, violating (P1).
	
	\smallskip
	\noindent\textbf{(P1) + (P3) $\Rightarrow$ $\neg$(P2).}
	The minimal AFDs are $f_1$ and $f_3$, so the only subset satisfying (P1) with $|\mathcal{S}| = 2$ is $\{f_1, f_3\}$. But $\sigma(f_3) = s_3 < s_2 = \sigma(f_2)$ and $f_2 \notin \mathcal{S}$, so $\mathcal{S}$ is not the global top-2, violating (P2).
	
	\smallskip
	\noindent\textbf{(P1) + (P2) $\Rightarrow$ $\neg$(P3).}
	To satisfy (P1), $\mathcal{S}$ must exclude the non-minimal $f_2$, so $\mathcal{S} \subseteq \{f_1, f_3\}$. The only subset of size $k = 2$ is $\{f_1, f_3\}$, which is not the global top-2 since $f_2$ has a higher score than $f_3$. No subset of size $k$ can therefore satisfy both (P1) and (P2).
	
	\smallskip
	In all three cases, at least one property is violated, completing the proof.
\end{proof}

\begin{figure}[t]
	\centering
	\begin{tikzpicture}[
		vertex/.style={draw, rounded corners=3pt, minimum width=1.8cm, minimum height=0.7cm, font=\small, align=center},
		edge label/.style={font=\scriptsize, fill=white, inner sep=2pt}
		]
		\node[vertex] (P1) at (0, 3.2) {(P1)\\Minimality};
		\node[vertex] (P2) at (-2.5, 0) {(P2)\\Global top-$k$};
		\node[vertex] (P3) at (2.5, 0) {(P3)\\Exact-$k$ output};
		
		\draw[thick] (P2) -- (P3) node[edge label, midway, below] {$\neg$(P1)};
		\draw[thick] (P1) -- (P3) node[edge label, midway, right] {$\neg$(P2)};
		\draw[thick] (P1) -- (P2) node[edge label, midway, left] {$\neg$(P3)};
		
	\end{tikzpicture}
	\caption{The Triangle Incompatibility (Theorem~\ref{thm:triangle}). Each edge indicates that enforcing the two connected properties implies the violation of the third.}
	\label{fig:triangle}
\end{figure}
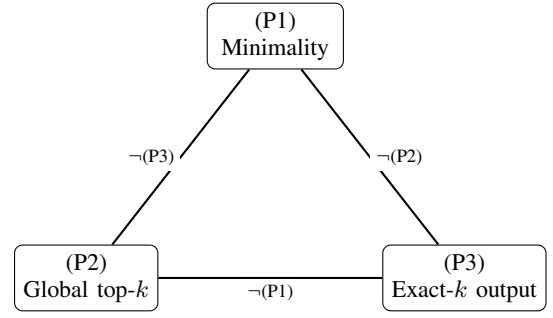

\subsection{Discussion}
Figure~\ref{fig:triangle} illustrates the structure of this incompatibility. The incompatibility is inherent in the problem definition, not a limitation of any particular algorithm: regardless of whether minimal AFDs are filtered before, after, or during ranking, the three properties cannot hold simultaneously.

The underlying reason is the absence of monotonicity. For exact FDs, $X \to A$ implies $X' \to A$ for any $X' \supset X$, so the larger dependency is redundant and minimality loses no information. For AFDs under $\mu^+$, no such implication exists: a non-minimal dependency $X' \to A$ can score strictly higher than any of its minimal subsets, so enforcing minimality may discard the globally highest-scoring dependency.

On the Adult dataset, $\{\mathtt{fnlwgt}\} \to \mathtt{sex}$ has $\mu^+ = 0.898$, while its strict superset $\{\mathtt{age}, \mathtt{fnlwgt}\} \to \mathtt{sex}$ achieves $\mu^+ = 0.986$, and $\{\mathtt{age}, \mathtt{fnlwgt}, \mathtt{relationship}\} \to \mathtt{sex}$ reaches $\mu^+ = 0.999$ (Section~\ref{sec:experiments}). The score increases strictly with LHS size, so $\mu^+$ is not monotone under LHS enlargement and the premise of Theorem~\ref{thm:triangle} holds for this measure.

We therefore enforce (P2) and (P3), returning the $k$ highest-scoring dependencies with exact output size, and drop (P1), prioritizing the strongest dependencies as measured by $\mu^+$ and aligning with the goal of identifying the most significant data quality issues.

\section{Baseline Algorithm} \label{sec:baseline}

\subsection{Overview}

We refer to our approach as \textbf{TALE} (\textit{T}op-k \textit{A}FD discovery via \textit{L}ogical Entropy). TALE-Base, the baseline variant presented in this section, evaluates all candidate AFDs exhaustively and is guaranteed to return the exact global top-$k$ result. TALE-Opt, introduced in Section~\ref{sec:optimized}, builds on this foundation with pruning and computation reuse.

TALE-Base enumerates candidates level by level, iterating over LHS sizes $\ell = 1, 2, \ldots, L$. At each level, it generates all $\binom{m}{\ell}$ attribute combinations of size $\ell$, pairs each with every attribute outside the LHS, and computes $\mu^+$ for each resulting candidate. A size-$k$ min-heap $\mathcal{H}$ maintains the current top-$k$ AFDs throughout. Each candidate requires a full scan of the relation, since the valid tuple set varies depending on which attributes contain NULL values. At the end of the enumeration, $\mathcal{H}$ contains the exact global top-$k$ result.

The following subsections describe candidate enumeration (Section~\ref{sec:enum}), top-$k$ heap maintenance (Section~\ref{sec:heap}), and complexity analysis (Section~\ref{sec:complexity}). Pseudocode of TALE-Base is given in Algorithm~\ref{alg:baseline}.

\begin{algorithm}[t]
	\caption{TALE-Base}
	\label{alg:baseline}
	\footnotesize
	\KwIn{Relation $r$ over schema $R = \{A_1, \ldots, A_m\}$, integer $k$, max LHS size $L$}
	\KwOut{Top-$k$ AFDs ranked by $\mu^+$ in descending order}
	
	$\mathcal{H} \leftarrow$ empty min-heap of capacity $k$, keyed on $\mu^+$\;
	\For{$\ell = 1$ \KwTo $L$}{
		\For{each $X \subseteq R$ with $|X| = \ell$}{
			\For{each $A \in R \setminus X$}{
				$r' \leftarrow \{t \in r \mid t[B] \neq \texttt{NULL}\ \text{for all}\ B \in X \cup \{A\}\}$\;
				\If{$|r'| \leq 1$}{\textbf{continue}\;}
				$s \leftarrow \mu^+(X \to A, r')$\tcp*{by Definition~\ref{def:muplus}}
				\If{$s = 1$}{\textbf{continue}\tcp*{exact FD}}
				\If{$|\mathcal{H}| < k$ \textbf{or} $s > \mathcal{H}.\mathrm{root}.\mathrm{score}$}{
					\If{$|\mathcal{H}| = k$}{remove root of $\mathcal{H}$\;}
					insert $(X \to A,\ s)$ into $\mathcal{H}$\;
				}
			}
		}
	}
	\Return entries of $\mathcal{H}$ sorted by $\mu^+$ descending\;
\end{algorithm}

\subsection{Candidate Enumeration}
\label{sec:enum}

TALE-Base enumerates all dependencies $X \to A$ with $X \subseteq R$, $A \in R \setminus X$, and $|X| \leq L$, filtering out exact FDs ($\mu^+ = 1$) after scoring. At level $\ell$, it generates all $\binom{m}{\ell}$ attribute subsets of size $\ell$ as candidate LHSs, pairing each with every attribute in $R \setminus X$ to yield $\binom{m}{\ell} \times (m - \ell)$ candidates at that level. The total number of candidates across all levels is
\begin{equation*}
	\sum_{\ell=1}^{L} \binom{m}{\ell} \cdot (m - \ell).
\end{equation*}

LHS combinations are generated in lexicographic order. Each $X$ is processed immediately: all $(m - \ell)$ candidates sharing that LHS are evaluated and passed to $\mathcal{H}$, after which $X$ is discarded. Only the current combination and $\mathcal{H}$ are kept in memory at any time.

For each candidate $X \to A$, TALE-Base makes a single pass over the relation, filtering tuples with NULL values in $X \cup \{A\}$, building the LHS-to-RHS frequency distribution, and computing pdep and $\mu^+$. Any candidate with $\mu^+ = 1$ is an exact FD and is excluded.

\subsection{Top-$k$ Heap Maintenance}
\label{sec:heap}

TALE-Base maintains a min-heap $\mathcal{H}$ of capacity $k$, keyed on $\mu^+$, whose root holds the current minimum score among the top-$k$ candidates. Two operations are performed on $\mathcal{H}$.

\textit{Insertion.} When a new candidate $X \to A$ is evaluated, it is inserted directly if $\mathcal{H}$ has fewer than $k$ entries. If $\mathcal{H}$ is full and the candidate's $\mu^+$ exceeds the root score, the root is replaced. Otherwise the candidate is discarded. Each insertion takes $O(\log k)$ time.

\textit{Threshold query.} The root score serves as a dynamic threshold $\tau$: any candidate with $\mu^+ \leq \tau$ cannot enter the top-$k$. In TALE-Base every candidate is evaluated before this comparison, so $\tau$ only saves the heap insertion cost. Its more important role is in TALE-Opt (Section~\ref{sec:optimized}), where it enables pruning entire branches without computing $\mu^+$.

At the end of the enumeration, $\mathcal{H}$ contains $\min(k, |\mathcal{F}|)$ AFDs. These are extracted in descending order of $\mu^+$ to produce the final result.

\subsection{Complexity Analysis}
\label{sec:complexity}

Let $m$ denote the number of attributes, $n$ the number of tuples, $L$ the maximum LHS size, and $k$ the desired output size.

\textit{Time complexity.} The total number of candidates is $C = \sum_{\ell=1}^{L} \binom{m}{\ell} \cdot (m - \ell)$. Computing $\mu^+$ for each candidate requires a single scan over the relation to filter NULL values and collect frequency statistics, costing $O(n)$. Each candidate also incurs an $O(\log k)$ heap operation. The overall time complexity is $O(C \cdot (n + \log k))$, which simplifies to $O(C \cdot n)$ in typical settings where $n \gg \log k$.

\textit{Space complexity.} Beyond the input relation, TALE-Base uses $O(n)$ working space for the per-candidate frequency maps (allocated and discarded per candidate, so only one set exists at any time), $O(k)$ for the min-heap, and $O(L)$ for the current LHS combination. The additional space complexity is therefore $O(n + k)$.

\section{Optimized Algorithm} \label{sec:optimized}
\subsection{Overview}

TALE-Opt reduces the exhaustive computation of TALE-Base through techniques at three levels.

\textit{The search level.} The independent combinatorial enumeration is replaced by an Apriori-style candidate generation scheme~\cite{DBLP:conf/vldb/AgrawalS94}. Candidates at level $\ell+1$ are produced by joining pairs from level $\ell$ that share a common prefix, so pruning decisions made at level $\ell$ reduce the candidates evaluated at deeper levels.

\textit{The computation level.} For a fixed $X$, the LHS hash values and NULL flags are identical across all candidates $X \to A$ and are computed once, eliminating $(m - \ell - 1)$ redundant passes over the LHS.

\textit{The pruning level.} Two complementary rules remove RHS attributes from consideration. Exact FD pruning exploits monotonicity: if $X \to A$ is exact ($\mu^+ = 1$), then $X' \to A$ is also exact for any $X' \supset X$, so all such descendants can be removed for RHS $A$. Upper bound pruning targets non-exact descendants via an optimistic upper bound on $\mu^+$ that decreases monotonically as $d_X$ grows. When the bound falls below the current heap threshold $\tau$, entire branches are pruned. The bound is strict under the no-NULL assumption; a heuristic extension for datasets with missing values is discussed in Section~\ref{sec:upperbound}.

\begin{algorithm}[t]
	\caption{TALE-Opt}
	\label{alg:optimized}
	\footnotesize
	\KwIn{Relation $r$ over schema $R = \{A_1, \ldots, A_m\}$, integer $k$, max LHS size $L$}
	\KwOut{Top-$k$ AFDs ranked by $\mu^+$ in descending order}
	
	$\mathcal{H} \leftarrow$ empty min-heap of capacity $k$, keyed on $\mu^+$\;
	$\tau \leftarrow 0$\;
	$\mathcal{L}_1 \leftarrow \bigl\{\{A_i\} : A_i \in R\bigr\}$\;
	\For{each $X \in \mathcal{L}_1$}{
		$S(X) \leftarrow R \setminus X$\;
	}
	\For{$\ell = 1$ \KwTo $L$}{
		\For{each $X \in \mathcal{L}_\ell$}{
			compute and cache $\textit{lhsHash}[i]$, $\textit{lhsValid}[i]$ for all tuples\;
			\For{each $A \in S(X)$}{
				compute $s \leftarrow \mu^+(X \to A, r)$ using cached LHS arrays\;
				\eIf{$s = 1$}{
					remove $A$ from $S(X)$\tcp*{FD pruning}
				}{
					\If{$s > \tau$}{
						insert or update $\mathcal{H}$ with $(X \to A, s)$\;
						$\tau \leftarrow \mathcal{H}.\mathrm{root}.\mathrm{score}$\;
					}
					\If{$\mu^+_{\mathrm{opt}}(X \to A, r) \leq \tau$}{
						remove $A$ from $S(X)$\tcp*{UB pruning}
					}
				}
			}
		}
		\If{$\ell < L$}{
			$\mathcal{L}_{\ell+1} \leftarrow \emptyset$\;
			\For{each ordered pair $X_1, X_2 \in \mathcal{L}_\ell$ sharing $(\ell{-}1)$-prefix}{
				$X' \leftarrow X_1 \cup X_2$\;
				$S(X') \leftarrow S(X_1) \cap S(X_2)$\;
				\If{$S(X') \neq \emptyset$}{
					add $X'$ to $\mathcal{L}_{\ell+1}$\;
				}
			}
		}
	}
	\Return entries of $\mathcal{H}$ sorted by $\mu^+$ descending\;
\end{algorithm}

Algorithm~\ref{alg:optimized} gives the pseudocode for the complete optimized algorithm TALE-Opt. The following subsections describe each technique in detail.

\subsection{Apriori-style Candidate Generation}
\label{sec:apriori}
In TALE-Base, each level $\ell$ generates candidates independently by enumerating all $\binom{m}{\ell}$ subsets of $R$ of size $\ell$, so a pruned candidate at level $\ell$ still has its supersets generated at level $\ell+1$. TALE-Opt replaces this with an Apriori-style scheme: candidates at level $\ell+1$ are derived from level $\ell$ by a join, so a candidate pruned at level $\ell$ produces no descendants.

Concretely, let $\mathcal{L}_\ell$ denote the surviving LHS candidates at level $\ell$. To generate $\mathcal{L}_{\ell+1}$, candidates in $\mathcal{L}_\ell$ are grouped by their first $\ell-1$ attributes (the \emph{prefix}); within each group, every pair whose last attributes differ is joined to produce a candidate of size $\ell+1$. For example, $\{A_1, A_3\}$ and $\{A_1, A_5\}$ in $\mathcal{L}_2$ join to produce $\{A_1, A_3, A_5\} \in \mathcal{L}_3$.

Apriori-style generation alone does not reduce the candidate space. It provides the structure through which pruning decisions propagate to descendants. The benefit appears when pruning rules remove candidates from $\mathcal{L}_\ell$, causing a cascading reduction in $\mathcal{L}_{\ell+1}$ and beyond.

\subsection{LHS Computation Reuse}
\label{sec:reuse}
For a given LHS $X$ of size $\ell$, TALE-Base evaluates $m - \ell$ candidates $X \to A$, each requiring an independent scan to compute the LHS hash and check for NULL values in $X$. Since $X$ is fixed across all these evaluations, the per-tuple LHS work is repeated unnecessarily.

TALE-Opt maintains two global arrays of size $n$: $\textit{lhsHash}[i]$, which stores the hash of the $i$-th tuple on the LHS attributes, and $\textit{lhsValid}[i]$, which records whether all LHS attributes of the $i$-th tuple are non-NULL. Both are computed once when a new $X$ is first processed and reused for all $m - \ell$ RHS evaluations. RHS-specific frequency counts are maintained separately for each $A$, but the tuple-side preprocessing on $X$ is shared.

When processing a subsequent RHS attribute $A$, the scan checks only whether $t[A]$ is NULL; the LHS hash and NULL status are read from the cached arrays, reducing the per-tuple LHS preprocessing cost from $O(\ell)$ to $O(1)$ for all but the first RHS. The asymptotic saving in LHS-side preprocessing is $O((m - \ell - 1) \cdot n \cdot \ell)$ per LHS candidate. The factor $\ell$ reflects the per-tuple cost saved by caching, while $(m - \ell - 1)$ counts the number of RHS evaluations that benefit from reuse.

\subsection{Pruning Framework and Exact FD Pruning}
\label{sec:pruning}

Each LHS candidate $X$ is associated with an \emph{RHS candidate set} $S(X) \subseteq R \setminus X$, tracking which RHS attributes still need evaluation for $X$ and its descendants. $S(X)$ is implemented as a bit vector of $m$ bits, where bit $i$ is set if and only if $A_i$ is still a candidate RHS. When a pruning rule determines that $A$ is unnecessary, it removes $A$ from $S(X)$.

\textit{Initialization.} At level~1, each single-attribute candidate $X = \{A_i\}$ is initialized with $S(X) = R \setminus \{A_i\}$.

\textit{Exact FD pruning.} After evaluating all RHS attributes in $S(X)$ for a given LHS $X$, every attribute $A$ for which $X \to A$ is exact ($\mu^+ = 1$) is removed from $S(X)$.

Exact FD satisfaction is closed under LHS enlargement: if $X \to A$ holds exactly, then $X' \to A$ holds exactly for any $X' \supset X$. Since exact FDs are excluded from the top-$k$ result, every descendant $X' \supset X$ with RHS $A$ would also be exact and thus excluded. Removing $A$ from $S(X)$ is therefore safe.

\textit{Propagation through Apriori join.}
When $X_1, X_2 \in \mathcal{L}_\ell$ are joined to produce $X' = X_1 \cup X_2 \in \mathcal{L}_{\ell+1}$, the RHS candidate set of $X'$ is initialized as
\begin{equation*}
	S(X') = S(X_1) \cap S(X_2).
\end{equation*}
If $A$ has been removed from $S(X_1)$ by exact FD pruning or by the upper-bound rule of Section~\ref{sec:upperbound}, then $X' \to A$ need not be evaluated: since $X' \supset X_1$, if $X_1 \to A$ is exact then $X' \to A$ is also exact, and if the upper bound for $X_1$ falls below $\tau$ then so does the upper bound for $X'$. The same applies to $S(X_2)$. If $S(X') = \emptyset$, $X'$ is discarded without being added to $\mathcal{L}_{\ell+1}$.

Intersection and emptiness testing both cost $O(\lceil m/w \rceil)$ where $w$ is the word size, and the cost can be reduced to $O(1)$ for schemas where the bit vector fits in one machine word.

\subsection{Optimistic Upper Bound Pruning}
\label{sec:upperbound}
Exact FD pruning targets attributes with $\mu^+ = 1$. A second rule targets attributes for which no descendant can score high enough to enter the top-$k$. It also operates by removing attributes from $S(X)$, and the propagation mechanism of Section~\ref{sec:pruning} applies without modification.

We first derive a strict upper bound under the no-NULL assumption, then discuss a heuristic extension for datasets with missing values.

\textit{Derivation under no-NULL assumption.} When there are no NULL values, every candidate operates on the full relation of $n$ tuples, and $\mathrm{pdep}(A, r)$ is a constant for each RHS attribute $A$. From Definition~\ref{def:muplus},
\begin{align*}
	\mu^+(X \to A, r) &= 1 - \rho(X \to A, r), \\
	\rho(X \to A, r) &= \frac{1 - \mathrm{pdep}(X \to A, r)}{1 - \mathrm{pdep}(A, r)} \cdot \frac{n-1}{n-d_X}.
\end{align*}
To derive an upper bound on $\mu^+$, we seek a lower bound on $\rho$.

Consider any descendant $X' \supset X$ such that $X' \to A$ is not an exact FD. Since $X \subset X'$, each distinct $X'$-value determines a unique $X$-value, so $d_{X'} \geq d_X$. It follows that $(n-1)/(n-d_{X'}) \geq (n-1)/(n-d_X)$, so the second factor of $\rho$ does not decrease.

For the first factor, $X' \to A$ being non-exact means that there exist at least two tuples $t_1, t_2 \in r$ with $t_1[X'] = t_2[X']$ but $t_1[A] \neq t_2[A]$. The scenario that minimizes $\rho$ (and thus maximizes $\mu^+$) is when the dependency is as close to exact as possible. This occurs when $X' \to A$ is violated by exactly one pair of tuples that share the same $X'$-value but disagree on $A$, while all other $X'$-values have a unique corresponding $A$-value. By Definition~\ref{def:pdep}, the contribution of this exception to $\mathrm{pdep}(X' \to A, r)$ drops by exactly $1/n$ compared to the fully exact case, giving $\mathrm{pdep}(X' \to A, r) = 1 - 1/n$ and $1 - \mathrm{pdep}(X' \to A, r) = 1/n$.

Substituting into $\rho$ gives the minimum value of $\rho$ over all non-exact descendants with $d_{X'}$ distinct LHS values,
\begin{align*}
	\rho_{\min}(d_{X'}) &= \frac{1}{n \cdot (1 - \mathrm{pdep}(A, r))} \cdot \frac{n-1}{n - d_{X'}} \\
	&= \frac{n-1}{n \cdot (1-\mathrm{pdep}(A, r)) \cdot (n-d_{X'})}.
\end{align*}
Since $d_{X'} \geq d_X$, we have $\rho_{\min}(d_{X'}) \geq \rho_{\min}(d_X)$. This leads to the following bound.

\begin{theorem}[Optimistic Upper Bound]
	\label{thm:upperbound}
	Assume the relation contains no NULL values, $d_X < n$, and $\mathrm{pdep}(A, r) < 1$. For any $X' \supset X$ such that $X' \to A$ is not an exact FD,
	\begin{align*}
		\mu^+(X' \to A, r) &\leq \mu^+_{\mathrm{opt}}(X \to A, r) \\
		&= 1 - \frac{n-1}{n \cdot (1-\mathrm{pdep}(A, r)) \cdot (n-d_X)}.
	\end{align*}
	Moreover, $\mu^+_{\mathrm{opt}}$ is monotonically decreasing in $d_X$.
\end{theorem}

\begin{proof}
	For any non-exact descendant $X' \supset X$ with $d_{X'} \geq d_X$,
	\begin{align*}
		\rho(X' \to A, r) &= \frac{1 - \mathrm{pdep}(X' \to A, r)}{1 - \mathrm{pdep}(A, r)} \cdot \frac{n-1}{n - d_{X'}} \\
		&\geq \frac{1}{n \cdot (1 - \mathrm{pdep}(A, r))} \cdot \frac{n-1}{n - d_{X'}} \\
		&\geq \frac{n-1}{n \cdot (1 - \mathrm{pdep}(A, r)) \cdot (n-d_X)}
	\end{align*}
	where the first inequality uses $\mathrm{pdep}(X' \to A, r) \leq 1 - 1/n$ (the tightest bound for a non-exact FD) and the second uses $d_{X'} \geq d_X$. Therefore $\mu^+(X' \to A, r) = 1 - \rho(X' \to A, r) \leq 1 - \rho_{\min}(d_X) = \mu^+_{\mathrm{opt}}(X \to A, r)$. Monotonicity follows because increasing $d_X$ decreases $(n - d_X)$, increasing $\rho_{\min}$ and decreasing $\mu^+_{\mathrm{opt}}$.
\end{proof}

\textit{Pruning rule.} This rule applies to non-exact candidates, complementing the exact FD pruning of Section~\ref{sec:pruning}. When $\mu^+ = 1$, the candidate is handled by exact FD pruning and $\mu^+_{\mathrm{opt}}$ is not evaluated. This also avoids the degenerate cases $d_X = n$ and $\mathrm{pdep}(A, r) = 1$ where the denominator of $\mu^+_{\mathrm{opt}}$ is zero. For candidates with $\mu^+ < 1$, if $\mu^+_{\mathrm{opt}}(X \to A, r) \leq \tau$, no non-exact descendant of $X$ with RHS $A$ can enter the top-$k$, and $A$ is removed from $S(X)$. Computing $\mu^+_{\mathrm{opt}}$ requires only $d_X$ and $\mathrm{pdep}(A, r)$, both available from the evaluation of $X \to A$, so the additional cost is $O(1)$ per candidate.

Since $\mu^+_{\mathrm{opt}}$ decreases monotonically in $d_X$, upper bound pruning is compatible with the propagation mechanism of Section~\ref{sec:pruning}: if $A$ is removed from $S(X_1)$ because $\mu^+_{\mathrm{opt}}(X_1 \to A, r) \leq \tau$, then for any $X' \supset X_1$, $d_{X'} \geq d_{X_1}$ implies $\mu^+_{\mathrm{opt}}(X' \to A, r) \leq \tau$, so $S(X') = S(X_1) \cap S(X_2)$ correctly excludes $A$.

\textit{Effectiveness across levels.} The upper bound depends on the ratio $(n-1)/[n \cdot (1-\mathrm{pdep}(A, r)) \cdot (n-d_X)]$. When $d_X$ is small relative to $n$, $\mu^+_{\mathrm{opt}}$ stays close to 1 and pruning has little effect. Pruning becomes increasingly aggressive at deeper levels as $d_X$ grows. 

\textit{Extension to datasets with NULL values.} When NULL values are present, different candidates may operate on different valid tuple sets, so both $n$ and $\mathrm{pdep}(A, r)$ vary across candidates and the bound of Theorem~\ref{thm:upperbound} is no longer guaranteed to hold strictly. We apply it heuristically, using the $n$ and $\mathrm{pdep}(A, r)$ values of the current candidate. On all tested datasets, the top-$k$ results with and without this heuristic are identical, suggesting the approximation is reliable in practice. For applications requiring absolute correctness guarantees, this pruning rule can be disabled without affecting the remainder of the algorithm.

\section{Experimental Evaluation}
\label{sec:experiments}

We evaluate TALE-Base and TALE-Opt on 41 real-world datasets spanning a wide range of sizes and dimensionalities, with the goal of assessing pruning effectiveness, scalability, and the semantic quality of the discovered dependencies.

\subsection{Experimental Setup}
\label{sec:setup}

All algorithms are implemented in Java (JDK 21) and executed on a workstation with an Intel Core i9-13900 CPU (2.00 GHz, 24 cores) and 64 GB of main memory, running Windows 11. All reported times are wall-clock times of single-threaded execution. A time limit of 48 hours is imposed. The executions that do not complete within this limit are marked with $>$48h.

We compare two variants: TALE-Base (Algorithm~\ref{alg:baseline}), which exhaustively evaluates all candidates, and TALE-Opt (Algorithm~\ref{alg:optimized}), which adds Apriori-style candidate generation, LHS computation reuse, exact FD pruning, and optimistic upper bound pruning. To our knowledge, no prior algorithm addresses the global top-$k$ AFD discovery problem as formulated in Definition~\ref{def:topk}. We do not compare against Mandros et al.~\cite{DBLP:conf/kdd/MandrosBV17} because the two methods address different problems: their method fixes the RHS and returns top-$k$ LHS subsets for a single target, whereas TALE returns a global top-$k$ across all attribute combinations. Moreover, their branch-and-bound strategy relies on an optimistic estimator specific to a fixed RHS; under $\mu^+$ and the global formulation, no analogous single bounding function exists across different RHS attributes, and the upper bound $\mu^+_{\mathrm{opt}}$ in Section~\ref{sec:upperbound} was derived specifically for the level-wise global search. Unless otherwise stated, $k = 20$ and $L = 5$.

We use 41 publicly available datasets from the UCI Machine Learning Repository\footnote{\url{https://archive.ics.uci.edu/}} and Kaggle\footnote{\url{https://www.kaggle.com/}}, covering a wide range of sizes ($|r|$ from 155 to 2,075,259), dimensionalities ($|R|$ from 7 to 109), and application domains. Dataset details and experimental results are given in Table~\ref{tab:overallPerformance}.

\textit{Handling of NULL values.} NULL values arise naturally during data collection, integration, and maintenance. We adopt ignore-NULL semantics for all experiments. An AFD $X \rightarrow A$ is evaluated only on tuples with non-NULL values on all attributes in $X \cup \{A\}$, consistent with recent work on AFD discovery~\cite{DBLP:journals/vldb/ParciakWHNPV25}.

\begin{table*}[t]
	\centering
	\caption{Overall performance of TALE-Base and TALE-Opt on 41 real-world datasets
		($k = 20$, $L = 5$, time limit 48 hours).
		\#ECN: number of evaluated candidates.
		PRatio: fraction of candidates pruned by TALE-Opt relative to TALE-Base.}
	\footnotesize
	\setlength{\tabcolsep}{4pt} 
	\begin{tabular}{|c||c|c||c c|c c||c|c|}
		\hline
		Dataset & $|r|$ & $|R|$  
		& \multicolumn{2}{c|}{TALE-Base} 
		& \multicolumn{2}{c||}{TALE-Opt} 
		& Speedup & PRatio\\
		\cline{4-7}
		& & 
		& Time[s] & \#ECN
		& Time[s] & \#ECN 
		& & \\
		\hline
		Abalone & 4177 & 9 & 1.06 & 1962 & 1.03 & 1763 & 1.03 & 10.14\% \\
		Adult & 32561 & 15 & 130.28 & 52080 & 91.91 & 48200 & 1.42 & 7.45\% \\
		AI4I 2020 & 10000 & 14 & 33.35 & 33306 & 17.02 & 29645 & 1.96 & 10.99\% \\
		Air Quality & 9357 & 15 & 63.42 & 52080 & 15.03 & 15262 & 4.22 & 70.70\% \\
		Bank Marketing & 45211 & 17 & 504.03 & 117028 & 383.53 & 112971 & 1.31 & 3.47\% \\
		Appliances Energy & 19735 & 29 & 11818.69 & 3550673 & 7539.76 & 3375632 & 1.57 & 4.93\% \\
		Bike Sharing & 17379 & 17 & 246.19 & 117028 & 151.54 & 104791 & 1.62 & 10.46\% \\
		Clickstream & 165474 & 14 & 533.57 & 33306 & 197.12 & 19883 & 2.71 & 40.30\% \\
		Coupon Recommendation& 12684 & 26 & 428.63 & 1778530 & 262.77 & 1502951 & 1.63 & 15.49\% \\
		Credit Card & 30000 & 25 & 8261.13 & 1386350 & 6059.55 & 1386350 & 1.36 & 0.00\% \\
		Dry Bean & 13611 & 17 & 279.78 & 117028 & 3.55 & 1374 & 78.81 & 98.83\% \\
		Flight & 1000 & 109 & 52135.27 & 12757176117 & 1468.64 & 51878348 & 35.50 & 99.59\% \\
		Gas Turbine & 36733 & 11 & 54.97 & 7007 & 16.66 & 1976 & 3.30 & 71.80\% \\
		Hepatitis & 155 & 20 & 2.47 & 333260 & 1.55 & 214449 & 1.59 & 35.65\% \\
		Household Power & 2075259 & 9 & 843.02 & 1962 & 694.19 & 1962 & 1.21 & 0.00\% \\
		Image Segmentation & 2100 & 20 & 71.11 & 333260 & 12.25 & 64256 & 5.81 & 80.72\% \\
		Incident Management & 141712 & 36 & 122223.16 & 13830012 & 23706.01 & 2197947 & 5.16 & 84.11\% \\
		Iranian Churn & 3150 & 14 & 7.38 & 33306 & 3.43 & 17616 & 2.15 & 47.11\% \\
		Letter Recognition & 20000 & 17 & 133.14 & 117028 & 107.71 & 117028 & 1.24 & 0.00\% \\
		Magic Gamma telescope & 19020 & 11 & 23.90 & 7007 & 4.42 & 1605 & 5.41 & 77.10\% \\
		MathE & 9546 & 8 & 0.55 & 952 & 0.34 & 655 & 1.62 & 31.20\% \\
		Metro Traffic & 48204 & 9 & 10.57 & 1962 & 7.50 & 1347 & 1.41 & 31.35\% \\
		MetroPT-3 & 1516948 & 17 & 37956.08 & 117028 & 23150.57 & 110242 & 1.64 & 5.80\% \\
		Mushroom & 8124 & 23 & 118.07 & 815166 & 79.50 & 759018 & 1.49 & 6.89\% \\
		Mushroom (secondary) & 61069 & 21 & 748.97 & 455679 & 393.97 & 265863 & 1.90 & 41.66\% \\
		NHANES & 2278 & 10 & 0.99 & 3810 & 0.83 & 3456 & 1.19 & 9.29\% \\
		Nursery & 12960 & 9 & 0.42 & 1962 & 0.48 & 1962 & 0.88 & 0.00\% \\
		Obesity & 2111 & 17 & 18.84 & 117028 & 16.08 & 114624 & 1.17 & 2.05\% \\
		Occupancy & 10129 & 19 & 175.24 & 239685 & 108.71 & 223768 & 1.61 & 6.64\% \\
		OpenFDA Drug & 1000 & 39 & 339.04 & 22812426 & 38.10 & 1699083 & 8.90 & 92.55\% \\
		Parkinsons & 5875 & 22 & 442.49 & 613690 & 28.39 & 64312 & 15.58 & 89.52\% \\
		Plista & 1001 & 63 & 2591.75 & 445323060 & 179.45 & 5413737 & 14.44 & 98.78\% \\
		Poker Hand & 1000000 & 11 & 349.72 & 7007 & 293.28 & 7007 & 1.19 & 0.00\% \\
		Product Classification & 35311 & 7 & 3.21 & 434 & 1.42 & 312 & 2.26 & 28.11\% \\
		Rice & 3810 & 8 & 0.66 & 952 & 0.47 & 882 & 1.40 & 7.35\% \\
		Spambase & 4601 & 58 & 63696.08 & 267555682 & 42423.48 & 235494705 & 1.50 & 11.98\% \\
		South German Credit & 1000 & 21 & 22.53 & 455679 & 14.21 & 386063 & 1.59 & 15.28\% \\
		Steel Industry & 35040 & 11 & 34.99 & 7007 & 19.19 & 6406 & 1.82 & 8.58\% \\
		Steel Plates & 1941 & 34 & 1734.17 & 9665282 & 63.07 & 551157 & 27.50 & 94.30\% \\
		Superconductivity & 21263 & 82 & $>$48h & 2244666114 & 6492.27 & 4366203 & $>$17.7 & 99.81\% \\
		WDBC & 569 & 32 & 334.61 & 6603744 & 174.71 & 6112655 & 1.92 & 7.43\% \\				
		Wine Quality & 4898 & 12 & 5.95 & 12276 & 3.55 & 9181 & 1.68 & 25.21\% \\
		\hline
	\end{tabular}
	\label{tab:overallPerformance}
\end{table*}

\subsection{Overall Performance}
\label{sec:overall}

Table~\ref{tab:overallPerformance} reports running times and evaluated candidate counts (\#ECN) for both variants, together with the speedup and pruning ratio (PRatio) of TALE-Opt, on all 41 datasets ($k = 20$, $L = 5$). TALE-Base completes on 40 of 41 datasets within the 48-hour limit. Only Superconductivity ($|R| = 82$) exceeds it; TALE-Opt completes on all 41 datasets, finishing Superconductivity in under two hours (6,492 seconds). The per-level speedup on Superconductivity grows from 1.25$\times$ at level 1 to 3.51$\times$ at level 2 and 25.43$\times$ at level 3, consistent with the theoretical behavior of the upper bound: as $\ell$ increases, $d_X$ grows toward $n$, $\mu^+_{\mathrm{opt}}$ decreases rapidly, and pruning triggers on a rapidly growing fraction of the candidate space.

PRatio ranges from 0\% to 99.81\% across the 41 datasets, driven primarily by data structure rather than dataset size. At one extreme, Dry Bean (98.83\%), Flight (99.59\%), Plista (98.78\%), Superconductivity (99.81\%), Steel Plates (94.30\%), and Incident Management (84.11\%) all exceed 84\% pruning. In these datasets, many LHS combinations produce $d_X$ values large relative to $n$, driving $\mu^+_{\mathrm{opt}}$ below $\tau$ at shallow levels. High dimensionality relative to $|r|$ amplifies this effect: for Flight ($|R| = 109$, $|r| = 1{,}000$), Plista ($|R| = 63$, $|r| = 1{,}001$), and OpenFDA Drug ($|R| = 39$, $|r| = 1{,}000$), $d_X$ approaches $n$ after adding even a single attribute to the LHS.

Five datasets (Credit Card, Household Power, Letter Recognition, Nursery, and Poker Hand) show zero pruning. These datasets contain no exact FDs among evaluated candidates, so exact FD pruning has no effect. The top-$k$ heap fills with low-scoring entries early, keeping $\tau$ near zero throughout the search and preventing the upper bound from triggering. Household Power and Poker Hand have narrow schemas ($|R| = 9$ and $11$), further limiting the candidate space.

Most datasets fall between these extremes, with PRatio between 10\% and 85\% and speedups between 1.2$\times$ and 6$\times$. Nursery is the only dataset where TALE-Opt is marginally slower than TALE-Base (0.48s vs.\ 0.42s): the Apriori generation overhead outweighs any pruning benefit on this small dataset with no prunable candidates. On all other datasets TALE-Opt matches or outperforms TALE-Base.

\textit{Summary.} TALE-Opt achieves a positive speedup on 40 of 41 datasets. Pruning is most effective where exhaustive evaluation is most expensive: high-dimensional schemas where $d_X$ grows rapidly with LHS size, causing upper bound pruning to activate aggressively at deeper levels and reducing the evaluated candidate count by one to three orders of magnitude.

\subsection{Effect of Search Depth}
\label{sec:effect_L}

\begin{figure}
	\centering
	\renewcommand{\thesubfigure}{}
	\subfigure[(a) Execution time]{
		\includegraphics[scale = 0.63]{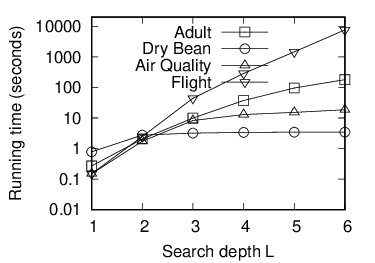}}
	\subfigure[(b) Heap threshold]{
		\includegraphics[scale = 0.63]{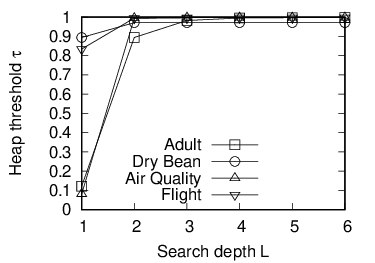}}
	\caption{Effect of search depth $L$ on running time and heap threshold ($k = 20$).}
	\label{fig:effectOfL}
\end{figure}

\begin{table}[t]
	\centering
	\caption{Top-$k$ list overlap between consecutive values of $L$ ($k = 20$).}
	\footnotesize
	\setlength{\tabcolsep}{4pt}
	\begin{tabular}{|c|c|c|c|c|}
		\hline
		$L \to L{+}1$ & Adult & Dry Bean & Air Quality & Flight \\
		\hline
		$1 \to 2$ & 5\%  & 5\%   & 0\%   & 0\%   \\
		$2 \to 3$ & 5\%  & 80\%  & 0\%   & 0\%   \\
		$3 \to 4$ & 15\% & 95\%  & 100\% & 0\%   \\
		$4 \to 5$ & 30\% & 100\% & 100\% & 100\% \\
		$5 \to 6$ & 70\% & 100\% & 100\% & 100\% \\
		\hline
	\end{tabular}
	\label{tab:overlap}
\end{table}

The parameter $L$ bounds the LHS size of discovered AFDs; a larger $L$ enlarges the search space and may improve result quality at the cost of longer running time. We evaluate TALE-Opt on four representative datasets with $L$ from 1 to 6, keeping $k = 20$ fixed. Figure~\ref{fig:effectOfL} reports execution time and heap threshold $\tau$; Table~\ref{tab:overlap} reports the overlap between consecutive top-$k$ lists.

Result quality converges well before $L = 5$ on all four datasets. On Dry Bean, $\tau$ stabilizes at $L = 2$ and the top-$k$ list reaches 100\% overlap at $L = 4$, with no further change through $L = 6$. On Air Quality, both $\tau$ and the top-$k$ list stabilize at $L = 3$, with 100\% overlap from $L = 3$ onward. On Flight, convergence occurs at $L = 4$, with 100\% overlap between $L = 4$ and $L = 5$. On Adult, $\tau$ continues to increase at every level but the increments diminish rapidly: $+0.773$ at $L = 2$, $+0.088$ at $L = 3$, $+0.013$ at $L = 4$, $+0.002$ at $L = 5$, and $+0.0005$ at $L = 6$; the overlap between $L = 5$ and $L = 6$ is 70\% and $\tau$ differs by less than 0.001, indicating near-convergence.

Running time grows with $L$ but at different rates across datasets. On Dry Bean, cumulative time at $L = 5$ is only 3.43 seconds and remains unchanged at $L = 6$, reflecting near-complete pruning at deeper levels. On Air Quality, time grows from 1.74 seconds at $L = 2$ to 15.26 seconds at $L = 5$. On Adult, time grows from 2.05 seconds at $L = 2$ to 96.13 seconds at $L = 5$ and 180.49 seconds at $L = 6$. Flight is the most expensive dataset due to its 109 attributes, with cumulative time reaching 1,469 seconds at $L = 5$ and 7,653 seconds at $L = 6$; the top-$k$ result does not change between these two levels, so the additional 6,184 seconds yield no benefit.

We adopt $L = 5$ as the default, since across all four datasets it either achieves the fully converged result or is within 0.001 of it in $\tau$, consistent with the finding of Mandros et al.~\cite{DBLP:conf/kdd/MandrosBV17} that the average optimal LHS size across 42 datasets is 4.0. Users may reduce $L$ for faster execution on wide schemas, or increase it when deeper completeness is needed.

\begin{figure}[t]
	\centering
	\renewcommand{\thesubfigure}{}	
	\subfigure[(a) Time vs. $n$ (MetroPT-3)]{
		\includegraphics[scale=0.63]{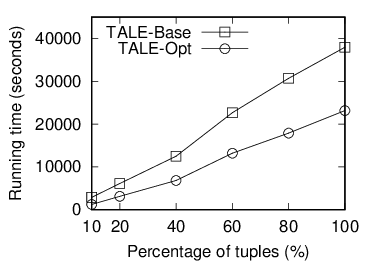}}
	\subfigure[(b) PRatio vs. $n$ (MetroPT-3)]{
		\includegraphics[scale=0.63]{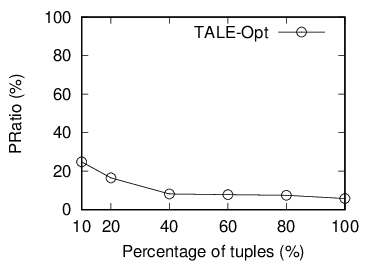}}
	\subfigure[(c) Time vs. $m$ (Spambase)]{
		\includegraphics[scale=0.63]{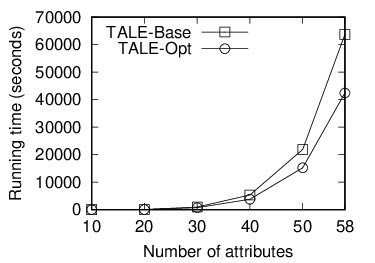}}
	\subfigure[(d) PRatio vs. $m$ (Spambase)]{
		\includegraphics[scale=0.63]{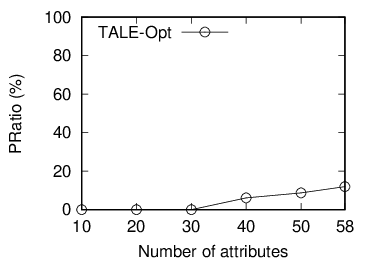}}
	\subfigure[(e) Time vs. $k$ (Adult)]{
		\includegraphics[scale=0.63]{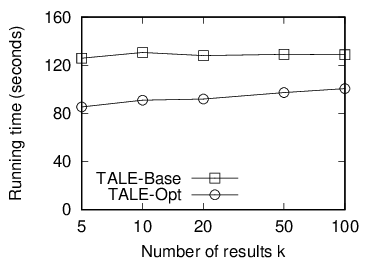}}
	\subfigure[(f) PRatio vs. $k$ (Adult)]{
		\includegraphics[scale=0.63]{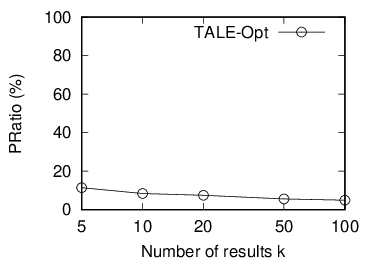}}
	\caption{Scalability of TALE-Base and TALE-Opt with respect to
		$n$, $m$, and $k$ ($k = 20$ for $n$ and $m$ experiments, $L = 5$ throughout).}
	\label{fig:scalability}
\end{figure}

\subsection{Scalability}
\label{sec:scalability}

We evaluate scalability with respect to $n$, $m$, and $k$. For $n$, we use MetroPT-3 ($|r| = 1{,}516{,}948$, $|R| = 17$), sampled at six fractions from 10\% to 100\%. For $m$, we use Spambase ($|r| = 4{,}601$, $|R| = 58$), selected for its high dimensionality. For $k$, we use Adult ($|r| = 32{,}561$, $|R| = 15$), where the effect of $k$ on pruning is clearly observable. Figure~\ref{fig:scalability} reports running time and PRatio for each experiment; $k = 20$ and $L = 5$ unless otherwise stated.

\textit{Varying $n$ (MetroPT-3).} TALE-Base grows from 2,869 seconds at 10\% to 37,956 seconds at 100\%, and TALE-Opt from 1,229 to 23,151 seconds. Both curves grow approximately linearly with $n$ (Figure~\ref{fig:scalability}(a)), as the candidate count depends only on $m$ and $L$, so total running time scales proportionally with $n$. PRatio decreases monotonically from 24.82\% to 5.80\% (Figure~\ref{fig:scalability}(b)): as $n$ grows, $d_X$ grows more slowly, keeping $(n - d_X)$ large and $\mu^+_{\mathrm{opt}}$ high, making the pruning condition harder to satisfy. The number of exact FDs also decreases with $n$ (from 8,129 at 10\% to 637 at 100\%), further reducing exact FD pruning.

\textit{Varying $m$ (Spambase).} Running time grows super-linearly with $m$ as the candidate space expands combinatorially. TALE-Base increases from 0.88 seconds at $m = 10$ to 63,696 seconds at $m = 58$, roughly $72{,}000\times$ (Figure~\ref{fig:scalability}(c)). TALE-Opt follows a similar trend but the gap widens from $m = 40$ onward, reaching approximately 21,000 seconds at $m = 58$. PRatio is zero for $m \leq 30$ (Figure~\ref{fig:scalability}(d)): Spambase contains no exact FDs at small schema sizes and the upper bound remains above $\tau$. Exact FDs appear at $m = 40$ (1,747 FDs), and PRatio rises to 11.98\% at $m = 58$ as both pruning rules become effective.

\textit{Varying $k$ (Adult).} TALE-Base is insensitive to $k$: running time stays between 125 and 131 seconds and \#ECN is fixed at 52,080 across all tested values (Figure~\ref{fig:scalability}(e)), as exhaustive evaluation covers all candidates regardless of $k$. TALE-Opt shows a mild increase from 85.39 seconds at $k = 5$ to 100.59 seconds at $k = 100$, with PRatio falling from 11.44\% to 4.90\% (Figure~\ref{fig:scalability}(f)): a larger $k$ requires a lower $\tau$, weakening the upper bound pruning condition. TALE-Opt remains faster than TALE-Base at every tested $k$, with speedups from 1.47$\times$ at $k = 5$ to 1.28$\times$ at $k = 100$.

\begin{table*}[t]
	\centering
	\caption{Effect of individual optimization components ($k = 20$, $L = 5$).
		Time is in seconds. PRatio is computed relative to the theoretical candidate count.}
	\footnotesize
	\setlength{\tabcolsep}{4pt}
	\begin{tabular}{|l||rrr|rrr|rrr|rrr|}
		\hline
		\multirow{2}{*}{Variant} 
		& \multicolumn{3}{c|}{Parkinsons} 
		& \multicolumn{3}{c|}{Air Quality} 
		& \multicolumn{3}{c|}{Adult} 
		& \multicolumn{3}{c|}{Letter Recognition} \\
		\cline{2-13}
		& Time(s) & \#ECN & PRatio
		& Time(s) & \#ECN & PRatio
		& Time(s) & \#ECN & PRatio
		& Time(s) & \#ECN & PRatio \\
		\hline
		TALE-Base  & 442.49 & 613,690 & —      & 63.42 & 52,080 & —      & 130.28 & 52,080  & —     & 133.14 & 117,028 & — \\
		TALE-A     & 435.93 & 613,690 & 0.00\% & 66.41 & 52,080 & 0.00\% & 130.65 & 52,080  & 0.00\% & 130.28 & 117,028 & 0.00\% \\
		TALE-A+R   & 225.07 & 613,690 & 0.00\% & 57.40 & 52,080 & 0.00\% & 104.71 & 52,080  & 0.00\% & 104.77 & 117,028 & 0.00\% \\
		TALE-A+R+E & 69.72  & 194,127 & 68.37\% & 44.54 & 42,124 & 19.12\% & 104.30 & 49,717 & 4.54\% & 106.53 & 117,028 & 0.00\% \\
		TALE-Opt   & 28.39  & 64,312  & 89.52\% & 15.03 & 15,262 & 70.70\% & 91.91  & 48,200 & 7.45\% & 107.71 & 117,028 & 0.00\% \\
		\hline
	\end{tabular}
	\label{tab:ablation}
\end{table*}

\subsection{Effect of Individual Optimization Components}
\label{sec:ablation}

Table~\ref{tab:ablation} reports running time, \#ECN, and PRatio for five variants on four datasets with distinct pruning characteristics: Parkinsons (PRatio 89.52\%), Air Quality (70.70\%), Adult (7.45\%), and Letter Recognition (0.00\%). Starting from TALE-Base, each variant adds one component: TALE-A adds Apriori-style candidate generation, TALE-A+R adds LHS computation reuse, TALE-A+R+E adds exact FD pruning, and TALE-Opt adds optimistic upper bound pruning.

\textit{Apriori-style generation.} TALE-A and TALE-Base evaluate the same number of candidates on all four datasets: Apriori-style generation alone does not reduce the candidate space but provides the propagation structure through which pruning decisions reach descendant LHS candidates. The running time difference is negligible.

\textit{LHS computation reuse.} TALE-A+R reduces running time substantially without changing \#ECN: on Parkinsons from 435.93 to 225.07 seconds (48.4\%); on Air Quality from 66.41 to 57.40 seconds (13.6\%); on Adult from 130.65 to 104.71 seconds (19.9\%); on Letter Recognition from 130.28 to 104.77 seconds (19.6\%). The larger saving on Parkinsons reflects its higher dimensionality ($|R| = 22$), where reusing LHS hash values across more RHS attributes yields greater benefit.

\textit{Exact FD pruning.} The effect of TALE-A+R+E depends on the number of exact FDs in each dataset. On Parkinsons (3,353 exact FDs), \#ECN drops from 613,690 to 194,127 (PRatio 68.37\%) and time from 225.07 to 69.72 seconds. On Air Quality, 18,406 exact FDs exist but their supersets cover fewer candidates, giving PRatio 19.12\% and time reduction from 57.40 to 44.54 seconds. On Adult (2,431 exact FDs), PRatio reaches only 4.54\% and time changes minimally. On Letter Recognition, no exact FDs exist and \#ECN is unchanged; the slight time increase (104.77 to 106.53 seconds) comes from BitSet maintenance overhead.

\textit{Optimistic upper bound pruning.} TALE-Opt provides the largest reduction where exact FD pruning alone leaves substantial candidates. On Parkinsons, \#ECN falls from 194,127 to 64,312, pushing PRatio to 89.52\% and time to 28.39 seconds. On Air Quality, upper bound pruning dominates: PRatio rises from 19.12\% to 70.70\% and time drops from 44.54 to 15.03 seconds (66.3\%). On Adult, PRatio reaches 7.45\% and time decreases from 104.30 to 91.91 seconds. On Letter Recognition, neither pruning rule triggers; TALE-Opt finishes in 107.71 seconds, close to TALE-A+R (104.77 seconds), where LHS reuse is the sole source of speedup.

All variants return identical top-$k$ results to TALE-Base on all four datasets, confirming that the pruning rules are safe.

\textit{Summary.} The four components target different sources of inefficiency. LHS reuse reduces per-candidate computation cost universally. Exact FD pruning is most effective when the data contains many exact FDs whose supersets can be eliminated without evaluation. Upper bound pruning dominates on datasets with strong approximate dependencies, activating when $\tau$ is high enough to cut off entire subtrees. Where neither pruning condition triggers, LHS reuse remains the primary source of speedup.

\begin{table*}[t]
	\centering
	\caption{Representative top-$k$ AFDs and their semantic significance
	across three domains ($k = 20$, $L = 5$). In WDBC, cp1 and cv3 
	abbreviate \textit{concave\_points1} and \textit{concavity3}; 
	in AI4I 2020, T\_wear, M\_fail, P\_temp, and R\_speed abbreviate 
	\textit{Tool\_wear}, \textit{Machine\_failure}, 
	\textit{Process\_temperature}, and \textit{Rotational\_speed}.}
	\scriptsize
	\setlength{\tabcolsep}{3pt}
	\begin{tabular}{|c|c|c|l|c|l|}
		\hline
		Dataset & $|r|$ & Rank & \multicolumn{1}{c|}{AFD} & $\mu^+$ &
		\multicolumn{1}{c|}{Semantic Significance} \\
		\hline
		\multirow{4}{*}{Dry Bean} & \multirow{4}{*}{13,611}
		& 4  & $(\mathit{roundness}) \to \mathit{Class}$ & 0.983
		& Single shape feature predicts bean variety \\
		\cline{3-6}
		& & 3  & $(\mathit{ConvexArea, Solidity}) \to \mathit{Class}$ & 0.983
		& Shape compactness jointly determines variety \\
		\cline{3-6}
		& & 12 & $(\mathit{Extent, Class}) \to \mathit{Area}$ & 0.972
		& Variety and fill ratio constrain bean size \\
		\cline{3-6}
		& & 14 & $(\mathit{Area, Extent}) \to \mathit{Compactness}$ & 0.971
		& Geometric consistency among morphological measurements \\
		\hline
		\multirow{4}{*}{WDBC} & \multirow{4}{*}{569}
		& 1  & $(\mathit{cp1}) \to \mathit{Diagnosis}$ & 0.763
		& Worst concave points strongly predict malignancy \\
		\cline{3-6}
		& & 4  & $(\mathit{perimeter1}) \to \mathit{Diagnosis}$ & 0.667
		& Tumor perimeter independently predicts diagnosis \\
		\cline{3-6}
		& & 15 & $(\mathit{Diagnosis, cv3}) \to \mathit{cp1}$ & 0.500
		& Diagnosis and concavity jointly constrain concave points \\
		\cline{3-6}
		& & 16 & $(\mathit{Diagnosis, cp1}) \to \mathit{cv3}$ & 0.500
		& Reverse dependency invisible to fixed-RHS methods \\
		\hline
		\multirow{4}{*}{AI4I 2020} & \multirow{4}{*}{10,000}
		& 1  & $(\mathit{T\_wear, M\_fail, TWF, HDF, PWF}) \to \mathit{OSF}$ & 0.995
		& Concurrent failures predict overstrain failure \\
		\cline{3-6}
		& & 2  & $(\mathit{Type, R\_speed, M\_fail, PWF, OSF}) \to \mathit{HDF}$ & 0.994
		& Machine type and speed predict heat dissipation failure \\
		\cline{3-6}
		& & 4  & $(\mathit{Torque, M\_fail, HDF, PWF, OSF}) \to \mathit{TWF}$ & 0.988
		& Torque with concurrent failures predicts tool wear failure \\
		\cline{3-6}
		& & 14 & $(\mathit{P\_temp, T\_wear, M\_fail, PWF}) \to \mathit{HDF}$ & 0.985
		& Temperature and tool wear indicate heat dissipation risk \\
		\hline
	\end{tabular}
	\label{tab:casestudy}
\end{table*}

\subsection{Case Study}
\label{sec:casestudy}

We examine the top-$k$ results on three datasets: Dry Bean (agricultural morphology), WDBC (medical diagnosis), and AI4I 2020 (industrial predictive maintenance). For each dataset, we select representative AFDs from the top-20 to illustrate patterns visible only through global discovery. Table~\ref{tab:casestudy} lists the selected AFDs and their semantic significance.

\textit{Dry Bean.} The top-20 results fall into two groups. Ranks 1--11 predict bean variety (\textit{Class}) from morphological shape features, all at $\mu^+ = 0.983$. Rank 4, $(\mathit{roundness}) \to \mathit{Class}$, is the most informative: a single shape descriptor nearly determines the variety, so a record whose roundness is inconsistent with its declared variety is a strong candidate for a labeling error. Rank 3 shows that \textit{ConvexArea} and \textit{Solidity} jointly achieve the same strength, characterizing how shape compactness varies across bean types. Ranks 12--20 show geometric consistency constraints, such as $(\mathit{Extent, Class}) \to \mathit{Area}$ and $(\mathit{Area, Extent}) \to \mathit{Compactness}$; violations here indicate measurement inconsistencies rather than labeling errors.

\textit{WDBC.} The Wisconsin Diagnostic Breast Cancer dataset contains 30 numerical features computed from cell nucleus images. Ranks 1--14 all predict \textit{Diagnosis} (benign or malignant) from individual or combined morphological features. The top dependency, $\mathit{concave\_points1} \to \mathit{Diagnosis}$, achieves $\mu^+ = 0.763$, identifying the worst concave points of the nucleus as the strongest single predictor. Rank 4 shows \textit{perimeter1} reaching $\mu^+ = 0.667$ independently, so tumor boundary length carries nearly comparable diagnostic signal. Ranks 15--20 shift direction: \textit{Diagnosis} now appears on the left-hand side, predicting morphological features such as \textit{concave\_points1} and \textit{concavity3}. A fixed-RHS method targeting \textit{Diagnosis} as the dependent attribute would suppress these entirely. Their appearance in the top-20 shows that diagnosis, combined with one morphological feature, strongly constrains others, relevant to consistency checking in clinical records.

\textit{AI4I 2020.} The AI4I 2020 dataset records sensor readings and failure modes for an industrial milling machine, with five binary failure indicators: TWF (tool wear failure), HDF (heat dissipation failure), PWF (power failure), OSF (overstrain failure), and RNF (random failure). The top-20 results span multiple right-hand sides: OSF appears in ranks 1 and 5--9, HDF in ranks 2--3 and 11--19, TWF in ranks 4, 10, and 13. Rank 1, $(\mathit{Tool\_wear}, \mathit{Machine\_failure}, \mathit{TWF}, \mathit{HDF}, \mathit{PWF}) \to \mathit{OSF}$ ($\mu^+ = 0.995$), shows that the co-occurrence of multiple concurrent failure conditions is a near-deterministic predictor of overstrain failure. Rank 2 shows that machine type, rotational speed, and concurrent failures jointly predict heat dissipation failure. Rank 14, $(\mathit{Process\_temp}, \mathit{Tool\_wear}, \mathit{Machine\_failure}, \mathit{PWF}) \to \mathit{HDF}$ ($\mu^+ = 0.985$), shows that process temperature combined with tool wear indicates heat dissipation risk without requiring knowledge of other failure states.

Across all three datasets, dependencies where the analyst's target attribute appears on the left-hand side, or where the most informative predictor is not the obvious candidate, appear consistently in the top-$k$ results. A fixed-RHS approach requires the analyst to commit to a target attribute before discovery begins; the global formulation removes this requirement and returns a unified ranking that can be inspected directly.

\section{Related work} \label{sec:related}

\subsection{Data Profiling}
Dependency discovery is a core task within data profiling~\cite{DBLP:journals/vldb/AbedjanGN15,DBLP:series/synthesis/2018Abedjan}, which extracts metadata and structural properties from datasets, including single-column statistics, unique column combinations, functional dependencies, and inclusion dependencies, with downstream applications in data quality assessment, schema understanding, and data integration.

Traditional profiling algorithms return complete result sets; when the number of attributes is large, the output can contain thousands of dependencies. A global top-$k$ ranking reduces this to a fixed-size result ordered by dependency strength.

\subsection{Functional Dependency Discovery}
Discovering the complete set of minimal non-trivial FDs from a relation instance has been studied extensively, and the algorithms can be grouped into three families.

Attribute-oriented algorithms traverse the attribute lattice to enumerate and validate candidates. TANE~\cite{DBLP:journals/cj/HuhtalaKPT99} performs a level-wise traversal using stripped partitions and applies inference rules for pruning; FUN~\cite{DBLP:journals/is/NovelliC01} and FD\_Mine~\cite{DBLP:journals/datamine/YaoH08} follow a similar enumeration with additional pruning conditions. DFD~\cite{DBLP:conf/cikm/AbedjanSN14} handles each RHS attribute independently through a depth-first walk.

Tuple-oriented algorithms derive dependencies from pairwise tuple comparisons. FDEP~\cite{DBLP:journals/aicom/FlachS99} and Dep-Miner~\cite{DBLP:conf/edbt/LopesPL00} compute agree sets or difference sets and extract FDs from these structures. FastFDs~\cite{DBLP:conf/dawak/WyssGR01} reduces cost by searching difference sets in a heuristic depth-first order. FSC~\cite{DBLP:journals/tkde/WanHWL24} addresses large-scale FD discovery by precomputing comparable pairs to avoid redundant tuple comparisons.

Hybrid algorithms combine elements of both families. HyFD~\cite{DBLP:conf/sigmod/PapenbrockN16} uses sampled tuple pairs to generate candidates and validates them through partition-based checks. FDHITS~\cite{DBLP:journals/pacmmod/BleifussPBSN24} formulates FD discovery as hitting-set enumeration and interleaves parallel exploration with efficient validation.

All of the above target exact FDs and produce complete result sets. Pyro~\cite{DBLP:journals/pvldb/0001N18} extends discovery to approximate FDs and UCCs under a user-specified error threshold, combining a separate-and-conquer strategy with sampling-based candidate detection, but still returns all dependencies below the threshold rather than ranking them. Our setting scores dependencies by $\mu^+$, does not fix the RHS, and returns only the top-$k$ results.

\subsection{AFD Measures}
Kivinen and Mannila~\cite{DBLP:journals/tcs/KivinenM95} defined several error measures for approximate FDs and studied their sample complexity. Among these, $g_3$ quantifies error through a partition-based formulation and was later adopted in TANE~\cite{DBLP:journals/cj/HuhtalaKPT99}, becoming widely used. Giannella and Robertson~\cite{DBLP:journals/is/GiannellaR04} proposed the $\tau$ measure based on probabilistic dependence (pdep), which captures how well the LHS determines the RHS but is sensitive to LHS-uniqueness. Mandros et al.~\cite{DBLP:conf/kdd/MandrosBV17} adopted a Shannon entropy-based measure with a bias correction derived from the permutation model.

Parciak et al.~\cite{DBLP:journals/vldb/ParciakWHNPV25} evaluated 14 AFD measures across reliability, sensitivity to LHS-uniqueness and RHS-skew, and computational cost, finding that $\mu^+$ and $\mathrm{RFI}'^+$ are the two best-performing measures. The authors recommend $\mu^+$ for practical use because $\mathrm{RFI}'^+$ is orders of magnitude more expensive to compute.

We adopt $\mu^+$ based on the findings of Parciak et al. The property that makes this choice non-trivial is that $\mu^+$ scores are comparable across different RHS attributes, unlike Shannon entropy-based measures whose normalization depends on the marginal entropy of each RHS.

\subsection{Top-$k$ AFD Discovery}
Mandros et al.~\cite{DBLP:conf/kdd/MandrosBV17} introduced the idea of returning the $k$ strongest approximate dependencies rather than all dependencies above a threshold, fixing a target attribute $A$ and searching for the $k$ subsets $X$ maximizing a corrected fraction of information $\hat{F}'(X; A)$ via branch-and-bound with an optimistic estimator. Follow-up work proved that maximizing the reliable fraction of information is NP-hard and derived a tighter bound~\cite{DBLP:conf/icdm/MandrosBV18,DBLP:journals/kais/MandrosBV20}.

Our formulation differs in that the RHS is not fixed: Mandros et al.\ search over $2^d$ LHS subsets for a single target, whereas we search all $X \to A$ combinations simultaneously. This requires a scoring function comparable across RHS attributes, which rules out the fraction of information whose normalization depends on $H(A)$, and invalidates their branch-and-bound strategy since the optimistic estimator assumes a fixed RHS. The question of whether minimality should be enforced arises naturally in the global setting; we resolve it negatively via the Triangle Incompatibility Theorem.

Wan et al.~\cite{DBLP:journals/corr/abs-2601-10130} rank exact FDs by redundancy count and develop pruning based on an upper bound on partition size; their scoring function is monotonic, so classical anti-monotonic pruning applies directly. Wei and Link~\cite{DBLP:journals/is/WeiL23} propose redundancy as a ranking criterion for meaningful FD discovery but do not integrate it into the search process for top-$k$ results. Fan et al.~\cite{DBLP:journals/pacmmod/FanHWX23,DBLP:journals/pacmmod/FanHXZ24} investigate top-$k$ discovery of general data quality rules incorporating relevance and diversity. None of these addresses the non-monotonicity of $\mu^+$, which is the central pruning challenge in our setting.

No existing method combines a global RHS-free formulation, a scoring function comparable across RHS attributes, and safe pruning under non-monotonic scoring. TALE addresses all three.

\section{Conclusion}
\label{sec:conclusion}

This paper studies global top-$k$ AFD discovery, where neither the LHS nor the RHS is fixed and the $k$ highest-scoring dependencies under $\mu^+$ are returned directly. The threshold-based paradigm that dominates existing work produces output of uncontrollable size, requires data-dependent threshold selection, and is sensitive to LHS dimensionality. The top-$k$ formulation avoids all three.

The Triangle Incompatibility Theorem shows that minimality, global top-$k$ ranking, and exact-$k$ output cannot simultaneously hold under any non-monotonic scoring function. This is an inherent property of the problem, not a limitation of any particular algorithm, and it provides a principled justification for dropping the minimality requirement.

TALE-Base guarantees the exact global top-$k$ result by evaluating all candidates level by level. TALE-Opt addresses the central challenge that $\mu^+$ is not monotone under LHS enlargement, which rules out classical anti-monotonic pruning. We derive an optimistic upper bound on $\mu^+$ that decreases monotonically with LHS size, recovering safe pruning; this bound is tracked per RHS attribute through bit-vectors and propagated across levels via Apriori-style candidate generation, complemented by LHS computation reuse and exact FD pruning. Experiments on 41 real-world datasets show pruning ratios reaching 99.81\% on high-dimensional datasets and speedups of TALE-Opt over TALE-Base up to 78.81$\times$. Case studies on three datasets confirm that the discovered dependencies are semantically interpretable and support practical data quality auditing.

\bibliographystyle{IEEEtran}
\bibliography{IEEEexample}


\vfill

\end{document}